\documentclass{cta-author}
\newcommand{\Am}{{\bf A}}
\newcommand{\xv}{{\bf x}}
\newcommand{\ev}{{\bf e}}
\newcommand{\muv}{\hbox{\boldmath$\mu$}}
\newcommand{\Sigmam}{\hbox{\boldmath$\Sigma$}}

\newcommand{\Sc}{\mathcal{S}} 
\newcommand{\eqdef}{\stackrel{\Delta}{=}}
\newcommand{\Hm}{{\bf H}}
\newcommand{\Sm}{{\bf S}}
\newcommand{\Nc}{ \mathcal{N}}
\newcommand{\Ac}{ \mathcal{A}}
\newcommand{\Kc}{ \mathcal{K}}
\newtheorem{prop}{Proposition}
\newcommand{\Pc}{\mathcal{P}}
\DeclareMathOperator*{\argmin}{arg\,min}
 
\newtheorem{definition}{Definition}
\newcommand{\sv}{{\bf s}}
\newcommand{\rv}{{\bf r}}
\usepackage{algorithm}
\usepackage{algpseudocode}

\usepackage{url}
\usepackage{balance}
\newtheorem{theorem}{Theorem}{}
\newtheorem{corollary}{Corollary}{}
{}
\usepackage{mhchem}
\usepackage{caption}
\usepackage{subcaption}
\begin{document}


\title{An information theoretic vulnerability metric for data integrity attacks on smart grids}
\author{Xiuzhen Ye$^{1*}$, I\~naki Esnaola$^{1,2}$, Samir M. Perlaza $^{3,2}$, and Robert F. Harrison$^{1}$}

\address{\add{1}{Dept. of Automatic  Control and Systems Engineering, University of Sheffield, Sheffield S1 3JD, UK}
\add{2}{Dept. of Electrical Engineering, Princeton University, Princeton, NJ 08544, USA}
\add{3}{INRIA, Sophia Antipolis 06902, France}
\email{Xye15@sheffield.ac.uk}}

\begin{abstract}
A novel metric that describes the vulnerability of the measurements in power systems to data integrity attacks is proposed. The new metric, coined vulnerability index (VuIx), leverages information theoretic measures to assess the attack effect on the fundamental limits of the disruption and detection tradeoff. The result of computing the VuIx of the measurements in the system yields an ordering of their vulnerability based on the level of exposure to data integrity attacks. This new framework is used to assess the measurement vulnerability of IEEE 9-bus and 30-bus test systems and it is observed that power injection measurements are overwhelmingly more vulnerable to data integrity attacks than power flow measurements. A detailed numerical evaluation of the VuIx values for IEEE test systems is provided.
\end{abstract}

\maketitle

\section{Introduction}
Supervisory Control and Data Acquisition (SCADA) systems and more recently advanced communication systems facilitate efficient, economic and reliable operation of power systems~\cite{GJ_PSanalysis_1994}. For instance, the communication system transmits the measurements to a state estimator that evaluates the operational status of the system accurately~\cite{AA_PSstateestimation_04}. However, the integration between the physical layer and the cyber layer exposes the system to cybersecurity threats. 
Cyber incidents highlight the vulnerability of power systems to sophisticated attacks. To ensure the security and reliability of power system operation, it is essential to quantitatively characterize the vulnerabilities of the system in order to set up appropriate security mechanisms~\cite{WZ_CN_13}. To that end, security metrics provide operationally meaningful vulnerability descriptors and identify the impact that security threats pose to the system. Moreover, security metrics enable operators to assess the defence mechanisms requirements to be embedded into cybersecurity policies, processes, and technology~\cite{JA_PearsonEducation_07}. For example, the Common Vulnerability Scoring System (CVSS) analysis Information Technology (IT) system~\cite{MS_IEEESecurityandPrivacy_06}. Typical security metrics for power systems focus on integrity, availability, and confidentiality as envisioned by the cybersecurity working group in the NIST Smart Grid interoperability panel~\cite{NISTIR_securityguideline_10}. 
System security objectives are categorized into system vulnerability, defence power, attack severity, and situations to develop security metrics in a systematic manner~\cite{PM_ACM_16}.
A cyberphysical security assessment metric (CP-SAM) based on quantitative factors is proposed to assess the specific security challenges of microgrid systems~\cite{VV_TSG_19}.

This fragmented landscape showcases a wide variety of metrics available that depend on the security services, threat characteristics, and system parameters. Remarkably, there is a lack of general data integrity vulnerability metrics for power systems. For instance, the impact of data injection attacks (DIAs)~\cite{LY_TISSEC_11} can be assessed with a wide variety of criteria that depend on the objectives of the attackers~\cite{CKKPT_SPM_12,MO_JSAC_13,IE_TSG_16}. A large body of literature addresses DIAs that compromise both the confidentiality and integrity of the information contained by the system measurements~\cite{NG_Stateestimation_21}. With the unprecedented data acquisition capabilities available to cyberphysical systems, the attackers can even learn the statistical structure of the system and incorporate the underlying stochastic process to launch the attacks~\cite{SE_TSG_19,YE_SGC_20}. DIAs that operate within a Bayesian framework by leveraging stochastic models of the system are studied in~\cite{SE_TSG_19} and~\cite{YE_SGC_20}. From the perspective of the operator, the introduction of stochastic descriptors opens the door to information theoretic quantifications of the measurement vulnerability.

In this paper, we propose a novel information theoretic metric to assess the vulnerability of measurements in power systems to data integrity attacks. Specifically, we characterize the fundamental information loss induced by data integrity attacks via mutual information and the stealthiness of the attack via Kullback-Leibler divergence. Our aim is to provide a metric that is grounded on fundamental principles, and therefore, inform the vulnerabilities of the measurements in the system to a wide range of threats. This is enabled by the use of information theoretic measures which characterize the amount of information acquired by the measurements in the system in fundamental terms.

The rest of the paper is organized as follows: In Section~\ref{sec_system_model}, we introduce a Bayesian framework with linearized dynamics for DIAs. Information theoretic attacks are presented in Section~\ref{sec_information_theoretic_attacks}. The vulnerability metric on information theoretic attacks is proposed in Section~\ref{sec_Vulnerability_Metric_on_Information_Theoretic_Attacks}. In Section~\ref{sec_vulnerability_of_measurements}, we characterize the vulnerability of measurements in uncompromised systems and propose an algorithm to evaluate the vulnerability of measurements. The vulnerability of measurements of the IEEE test systems is presented in Section~\ref{sec_simulation}. The paper concludes in Section~\ref{sec_conclusion}.

The main contributions of this paper follow: 
(1) A notion of vulnerability for the measurements in the system is proposed. The proposed notion is characterized by the information theoretic cost induced by random attacks. Specifically, mutual information and KL divergence are used to construct a quantitative measure of vulnerability.
(2) The vulnerability assessment of the measurements is posed as a minimization problem and closed-form expressions are obtained for the case in which the initial state of the system is uncompromised.	
(3) An algorithm that computes the proposed vulnerability indices for general state estimators in power systems is proposed.
(4) The proposed framework is numerically evaluated in IEEE 9-bus and 30-bus test systems to obtain qualitative characterizations of the vulnerability of the measurements in the systems. 

\textbf{Notation:} We denote the number of state variables on a given system by $n$ and the number of the measurements by $m$. The set of positive semidefinite matrices of size $n\times n$ is denoted by $S_{+}^n$. The $n$-dimensional identity matrix is denoted as $\textbf{I}_n$. 
For a matrix $\Am \in \mathbb{R}^{m \times n}$, then $(\Am)_{ij}$ denotes the entry in row $i$ and column $j$ and $\textnormal{diag}(\Am)$ denotes the vector formed by the diagonal entries of $\Am$.
The elementary vector $\ev_i\in\mathbb{R}^n$ is a vector of zeros with a one in the $i$-th entry. Random variables are denoted by capital letters and their realizations by the corresponding lower case, e.g., $x$ is a realization of the random variable $X$. Vectors of $n$ random variables are denoted by a superscript, e.g., $X^n=(X_1, \ldots, X_n)^{\sf{T}}$ with corresponding realizations denoted by $\xv$.
Given an $n$-dimensional vector $\muv \in \mathbb{R}^n$ and a matrix $\Sigmam \in S_{+}^n$, we denote by $\mathcal{N} (\muv, \Sigmam)$ the multivariate Gaussian distribution of dimension $n$ with mean $\muv$ and covariance matrix $\Sigmam$.
The mutual information between random variables $X$ and $Y$ is denoted by $I(X;Y)$ and the Kullback-Leibler (KL) divergence between the distributions $P$ and $Q$  is denoted by $D(P\| Q)$.

\section{System model}\label{sec_system_model}
\subsection{Observation Model}
In a power system the state vector $\xv \in{\mathbb{R}^n}$ that contains the voltages and phase angles at all the buses describes the operational state of the system. State vector $\xv$ is observed by the acquisition function $F: {\mathbb{R}^n} \rightarrow {\mathbb{R}^m}$. A linearized observation model is considered for state estimation, which yields the observation model as 
\begin{equation}\label{eq:obs_noattack}
	Y^m  = \Hm\xv+Z^m,
\end{equation}
where $\Hm \in {\mathbb{R}^{m \times n}}$ is the Jacobian of the function $F$ at a given operating point and is determined by the parameters and topology of the system. The vector of measurements $Y^m$ is corrupted by additive white Gaussian noise introduced by the sensors~\cite{GJ_PSanalysis_1994},~\cite{AA_PSstateestimation_04}. The noise vector $Z^m$ follows a multivariate Gaussian distribution, that is,
\begin{equation} \label{Z_distribution}
Z^m \sim \mathcal{N}(\textbf{0},\sigma^2 \textrm{\textbf{I}}_m),
\end{equation}
where $\sigma^2$ is the noise variance.
 
In a Bayesian estimation framework, the state variables are described by a vector of random variables $X^n$ with a given distribution. In this study, we assume $X^n$ follows a multivariable Gaussian distribution~\cite{GE_TSG_18} with zero mean and covariance matrix $\Sigmam_{X\!X} \in S_{+}^n$, that is,
\begin{equation}\label{Sigma_xx}
	X^n \sim \mathcal{N}(\textbf{0},\Sigmam_{X\!X}).
\end{equation}
From~(\ref{eq:obs_noattack}), it follows that the vector of measurements is with zero mean and covariance matrix $\Sigmam_{Y\!Y} \in S_{+}^m$, that is, 
\begin{equation}\label{Y_Sigma_YY}
Y^m \sim  \mathcal{N}(\textbf{0},\Sigmam_{Y\!Y}),
\end{equation}
where 
\begin{equation}\label{20220706_1}
\Sigmam_{Y\!Y} \eqdef \Hm\Sigmam_{X\!X}\Hm^{\sf{T}}+\sigma^2 \textrm{\textbf{I}}_m.
\end{equation}
\subsection{Attack Setting}
Let us denote the measurements corrupted by the malicious attack vector $A^m \in \mathbb{R}^m$ as $Y_A^m$, that is,
\begin{equation}\label{obs_attack}
	Y_A^m  = \Hm X^n+Z^m + A^m,
\end{equation}
where $A^m \sim P_{A^m}$ is the distribution of the random attack vector $A^m$. With a fixed covariance matrix $\Sigmam_{A\!A}$, when the additive disturbance to the system, that is, $Z^m + A^m$ follows a multivariate Gaussian distribution, the mutual information between the state variables $X^n $ and the compromised measurements $Y_A^m$ denoted by $I(X^n;Y^m_A)$ is minimized~\cite{SI_TIT_13}. Hence, from the L\'{e}vy-Cram\'{e}r decomposition theorem~\cite{Levy_levycramertheorem_35,cramer_levycramertheorem_35}, it holds that for the sum $Z^m + A^m$ to be Gaussian, given that $Z^m$ satisfies \eqref{Z_distribution}, then, $A^m$ must be Gaussian. Therefore, in the following, we assume that 
\begin{equation}\label{eq:Gauss_attack}
	A^m \sim \mathcal{N} (\textbf{0}, \Sigmam_{A\!A}),
\end{equation}
where $\textbf{0} = (0,0,\ldots,0)$ and $\Sigmam_{A\!A} \in \Sc_+^m$ are the mean vector and the covariance matrix of the random attack vector $A^m$. 
The assumption in~\eqref{eq:Gauss_attack} is further discussed in Section~\ref{sec_information_theoretic_attacks}. Consequently, the vector of compromised measurements $Y_A^m$ follows a multivariate Gaussian distribution with zero mean and covariance matrix $\Sigmam_{Y_A\!Y_A} \in \Sc_+^m$, that is,
\begin{equation}\label{eq:obs_attack} 
	Y^m_A  \sim \mathcal{N} (\textbf{0},\Sigmam_{Y_A\!Y_A}),
\end{equation}
with 
\begin{equation}\label{eq:Sigma_YA} 
\Sigmam_{Y_A\!Y_A} \eqdef \Hm\Sigmam_{X\!X}\Hm^{\sf{T}} + \sigma^2 \textrm{\textbf{I}}_m + \Sigmam_{A\!A}. 
\end{equation}

\vspace{-5mm}
\section{Information Theoretic Attacks}\label{sec_information_theoretic_attacks}
The aim of the attack is two fold. Firstly, the attack aims to disrupt the state
estimation procedure. Secondly, it aims to stay undetected. For the first objective, we minimize the mutual information between the vector of state variables $X^n$ in~\eqref{Sigma_xx} and the vector of compromised measurements $Y_A^m$ in~\eqref{obs_attack}, that is, $I(X^n;Y_A^m)$. 
In the other words, the attack yields less information about the state variables contained by the compromised measurements.
The stealthy constraint in the second objective is captured by the Kullback Leibler (KL) divergence between the distribution $P_{Y^m_A}$ in~\eqref{obs_attack} and the distribution $P_{Y^m}$ in~\eqref{eq:obs_noattack}, that is, $D(P_{Y_A^m}\|P_{Y^m})$. The Chernoff-Stein Lemma~\cite{book_EIT} states that the minimization of KL divergence leads to the minimization of the asymptotic detection probability. 

The following propositions characterize mutual information and KL divergence with Gaussian state variables and attacks, respectively~\cite[Prop. 1, 2]{SE_TSG_19}.
\begin{prop}\label{Prop_MI}
	The mutual information between the vectors of random variables $X^n$ in~\eqref{Sigma_xx} and $Y_A^m$ in~\eqref{eq:obs_attack} is 
	\begin{eqnarray}\label{eq_MI}
	I(X^n;Y^m_A) 
		= \dfrac{1}{2}\log \dfrac{|\Sigmam_{X\!X}||\Sigmam_{Y_A\!Y_A}|}{|\Sigmam|},
	\end{eqnarray}
where the matrices $\Sigmam_{X\! X}$ and $\Sigmam_{Y_A\!Y_A}$ are in~\eqref{Sigma_xx} and~\eqref{eq:Sigma_YA}, respectively; and the matrix $\Sigmam$ is the covariance matrix of the joint distribution of $X^n$ and $Y_A^m$, that is, $(X^n;Y_A^m) \sim \Nc (\textbf{0}, \Sigmam)$ with
	\begin{eqnarray}\label{Sigmam_def}
	\Sigmam = \left(\begin{matrix} \Sigmam_{X\!X} & \Sigmam_{X\!X}\Hm^{\sf T} \\ \Hm\Sigmam_{X\!X}  & \Hm\Sigmam_{X\!X}\Hm^{\sf T} + \sigma^2\textbf{I}_m + \Sigmam_{A\!A} \end{matrix} \right),
\end{eqnarray} 
	where the real $\sigma \in \mathbb{R}_+$ is in~\eqref{Z_distribution}; the matrices $\Hm$ and $\Sigmam_{A\!A}$ are in~\eqref{eq:obs_noattack} and~\eqref{eq:Gauss_attack}, respectively.
\end{prop}

\begin{prop}\label{Prop_KL}
	The KL divergence between the distribution of random vector $Y_A^m$ in~\eqref{eq:obs_attack} and the distribution of random vector $Y^m$ in~\eqref{Y_Sigma_YY} is 
	\begin{align}\label{eq_KL} 
	 D(P_{Y_A^m}\|P_{Y^m}) =  \frac{1}{2}\left( \log\frac{\left| \Sigmam_{Y\!Y} \right|}{\left|\Sigmam_{Y_A\!Y_A}\right|} -m +  \textnormal{tr} \left(\Sigmam_{Y\!Y}^{-1} \Sigmam_{Y_A\!Y_A} \right)  \right),
	\end{align}
	where the matrices $\Sigmam_{Y\!Y}$ and $\Sigmam_{A\!A}$ are in~\eqref{20220706_1} and \eqref{eq:Gauss_attack}, respectively.
\end{prop}

Particularly, the attack construction is proposed in the following optimization problem~\cite{SE_TSG_19,YE_SGC_20}:
\begin{equation}\label{eq:stealth_opt}
	\min_{P_{A^m}} I(X^n;Y^m_A)+ \lambda D(P_{Y_A^m}\|P_{Y^m}),
\end{equation}
where $\lambda \in \mathbb{R}_+$ is the weighting parameter that determines the tradeoff between mutual information and KL divergence. 
Note that the optimization domain in~\eqref{eq:stealth_opt} is the set of $m$-dimensional Gaussian multivariate distributions. The optimal Gaussian attack for $\lambda \geq 1$ as a solution to~\eqref{eq:stealth_opt} is given by~\cite{SE_TSG_19} 
\begin{equation}\label{opt_stealth}
P_{A^m}^* \sim\Nc(\mathbf{0},{\lambda^{-1/2}}\Hm\Sigmam_{X\!X}\Hm^{\sf T}).
\end{equation}
Note that the attack realizations from~\eqref{opt_stealth} are nonzero with probability one, that is, $\mathbb{P}[\left| \textnormal{supp}({A^m})\right |= m] = 1$, where
\begin{equation}
\textnormal{supp}({A^m})\eqdef\left\{i:\mathbb{P}\left[A_i=0\right]=0\right\}.
\end{equation}
The attack implementation requires access to the sensing infrastructure of the industrial control system (ICS) operating the power system. 
For that reason, the attack construction incorporates the sparsity constraint by limiting the optimization domain over the attack vector $A^m$ in~\eqref{obs_attack} to the distributions with cardinality of the support satisfying $|\textnormal{supp}({A^m})|=k\leq m$, that is,
\begin{equation}
\label{EqSparcityDomain}
\Pc_k\eqdef\bigcup_{i=1}^k\left \{{A^m}\sim \Nc(\mathbf{0},\bar{\Sigmam}):\left| \textnormal{supp}({A^m})\right | =i\right \}.
\end{equation}
The resulting attack construction with the sparsity constraints is
\begin{equation}\label{eq:sparsestealth_opt} 
	\min_{\Pc_k} I(X^n;Y^m_A)+ \lambda D(P_{Y_A^m}\|P_{Y^m}).
\end{equation}
The following theorem provides the optimal single sensor attack construction.
\begin{theorem}\cite[Th. 1]{YE_SGC_20}\label{Theorem_singlesensor} 
	The solution to the sparse stealth attack construction problem in~\eqref{eq:sparsestealth_opt}  for the case $k=1$ is 
	\begin{equation}\label{26}
		\bar{\Sigmam}^*=v\ev_{i}\ev_{i}^{\sf T},
	\end{equation}
	where
		\begin{align}
			i&= \argmin_{j \in \{1,2,\ldots,m\}} \left\{\left(\Sigmam_{Y\!Y}^{-1}\right)_{jj}\right\} \label{eq:opt_i}\\
			v&=- \frac{\sigma^2}{2} + \frac{1}{2}\left( \sigma^4 -\frac{4 (\underbar{$w$}\sigma^2 -1)}{\lambda \underbar{$w$}^2} \right)^{\frac{1}{2}} \label{eq:opt_v}
		\end{align}
	with $\underbar{$w$}\eqdef (\Sigmam_{Y\!Y}^{-1})_{ii}$.
\end{theorem}

\section{Vulnerability Metric on Information Theoretic Attacks}\label{sec_Vulnerability_Metric_on_Information_Theoretic_Attacks}
\subsection{Attack Structure with Sequential Measurement Selection}
To assess the impact of the attacks to different measurements, we model the entries of the random attack vector $A^m$ with independency, that is, 
\begin{equation}
	\label{EqIndepPA}
	P_{A^m}=\prod_{i=1}^m P_{A_i},
\end{equation} 
where $A_i$ is the $i$-th entry of $A^m$ and for all $i \in \lbrace 1,2, \ldots, m\rbrace$, the distribution $P_{A_i}$ is Gaussian with zero mean and variance $v \in \mathbb{R}_+$, that is, $A_i \sim \Nc (0,v)$.  
Consider that $k$ sensors have been attacked with $k \in \{0,1,2,\ldots,m-1\}$ and let the covariance matrix of the corresponding attack vector $A^m$ in~\eqref{obs_attack} be
\begin{equation}\label{20220310_1}
\Sigmam \in \Sc_k,
\end{equation}
where $\Sc_k$ is the set of $m$-dimensional positive semidefinite matrix with $k$ positive entries in the diagonal, that is,
\begin{equation}\label{20220609_1}
\Sc_k \eqdef \{\Sm \in \Sc_+^m:\|\textnormal{diag}(\Sm)\|_0 = k\}.
\end{equation}
Let the set of measurements that have not been compromised be
\begin{eqnarray}\label{set_K_o}
	\Kc_o \eqdef   \lbrace i \in \lbrace 1,2, \ldots, m \rbrace: (\Sigmam)_{ii} = 0 \rbrace,
\end{eqnarray} 
where $(\Sigmam)_{ii}$ is the entry of $\Sigmam$ in row $i$ and column $i$.
The sequential measurement selection imposes the following structure in the covariance matrix of the attack vector in~\eqref{eq:Gauss_attack}:
\begin{equation}\label{20220312_1}
\Sigmam_{A\!A} = \Sigmam+v\ev_i \ev_i^{\sf T},
\end{equation}
where $i\in\Kc_o$ and $v\in\mathbb{R}_+$. From equation~\eqref{20220312_1}, the cost function $f: \Sc_k \times \mathbb{R}_+ \times \mathbb{R}_+ \times \Kc_o \rightarrow \mathbb{R}_+$ defined by adding \eqref{eq_MI} and~\eqref{eq_KL} is as follows:
\begin{align} \label{MI_f_k}
	&f(\Sigmam, \lambda,v, i) \\
	\eqdef & I(X^n;Y^m_A) +\lambda D(P_{Y_A^m}\|P_{Y^m}) \label{20220924_1}\\
	= & \dfrac{1}{2}\log \dfrac{|\Sigmam_{X\!X}||\Sigmam_{Y_A\!Y_A}|}{|\Sigmam|} \label{20220924_2}\\
	\nonumber
	&+ \frac{1}{2}\lambda\left( \log\frac{\left| \Sigmam_{Y\!Y} \right|}{\left|\Sigmam_{Y_A\!Y_A}\right|} -m +  \textnormal{tr} \left(\Sigmam_{Y\!Y}^{-1} \Sigmam_{Y_A\!Y_A} \right)  \right)\\
	=& \dfrac{1}{2}\log \dfrac{ |\Sigmam_{Y_A\!Y_A}|}{|\sigma^2\textbf{I}_m + \Sigmam_{A\!A}|}  +  \frac{1}{2}\lambda\left( \log\frac{\left| \Sigmam_{Y\!Y} \right|}{\left|\Sigmam_{Y_A\!Y_A}\right|}  +  \textnormal{tr} \left(\Sigmam_{Y\!Y}^{-1} \Sigmam_{A\!A} \right)  \right),\label{20220924_3}\\
		\nonumber
    = & \frac{1}{2}(1-\lambda)\textnormal{log}  \left| \Sigmam_{Y\!Y} + \Sigmam + v\ev_i \ev_i^{\sf T} \right| - \frac{1}{2}\textnormal{log}  \left|\Sigmam + v\ev_i \ev_i^{\sf T} +\sigma^2 \textbf{I}_m\right| \\
	&+\frac{1}{2} \lambda \left( \textnormal{tr}\left(\Sigmam_{Y\!Y}^{-1}\left(\Sigmam + v\ev_i \ev_i^{\sf T}\right) \right) + \log \left| \Sigmam_{Y\!Y}\right|\right)\label{20220915_2},
\end{align} 
where the inequality in~\eqref{20220924_2} holds from plugging in~\eqref{eq_MI} and~\eqref{eq_KL} into~\eqref{20220924_1}; the equality in~\eqref{20220924_3} follows from cancelling $|\Sigmam_{X\!X}|$ in the first term~\cite[Sec. 14.17]{Seber_book} and noting that $\Sigmam_{Y_A\!Y_A} = \Sigmam_{Y\!Y} + \Sigmam_{A\!A}$ in~\eqref{eq:Sigma_YA}; and the equality in~\eqref{20220915_2} holds from plugging in~\eqref{20220312_1} into~\eqref{20220924_3}.
\subsection{Information theoretic vulnerability of a measurement}

We propose a notion of vulnerability that is linked to the information theoretic cost function proposed in \cite{SE_TSG_19} to characterize the disruption and detection tradeoff incurred by the attacks. Taking the state of the system with $k$ compromised measurements as the baseline, we quantify the vulnerability of measurement $i$ in terms of the cost decrease induced by attacking the correspond sensor $i$ that obtains the measurement, with $i\in \Kc_o$. 
In the following, we define the vulnerability of a measurement.
\begin{definition}
	\label{def:vul_unc}
	The function $\Delta: \Sc^m_+ \times \mathbb{R}_+ \times \mathbb{R}_+ \times \Kc_o \rightarrow \mathbb{R}_+$, where $\Kc_o$ is in~\eqref{set_K_o}, defines the {\it vulnerability of measurement $i$} in the following form:
	\begin{equation}
	\label{def_Delta_i}
	\Delta(\Sigmam, \lambda,v, i) \eqdef f(\Sigmam, \lambda,v, i) -f(\Sigmam, \lambda,0, i),
\end{equation}
	where the function $f$ is defined in~\eqref{MI_f_k}.
\end{definition}
Note that the attacker aims to minimize~\eqref{MI_f_k} by choosing an index $i$ and a variance $v$, and therefore, the definition above implies that given that $k$ sensors in $\{1,2,\ldots,m\}\setminus\Kc_o$ are already attacked in the system, the most vulnerable measurement is obtained by solving the following minimization problem
\begin{equation}\label{20020310_13}
\min_{i \in \Kc_o} \Delta(\Sigmam, \lambda, v, i),
\end{equation}
where $\Kc_o$ is defined in~\eqref{set_K_o}.

\section{Vulnerability of Measurements}\label{sec_vulnerability_of_measurements}
\subsection{Vulnerability analysis of uncompromised systems}\label{Vulnerability analysis on single sensor attack}
We first consider the case in which no sensors are under attacks, that is, $k=0$ and the following holds 
	\begin{align}
		\Sigmam = & \ \textbf{0},\label{20220915_1}  \\
		\Kc_o =  & \  \lbrace 1,2, \ldots, m \rbrace.\label{20220915_3}
	\end{align}

The attacker selects a single sensor and corrupts the corresponding measurement with a given budget $v\leq v_0$. We quantify the vulnerability of measurement $i$ in terms of $\Delta(\Sigmam, \lambda, v, i)$ defined in~\eqref{def_Delta_i}.

For the uncompromised system case, the optimization problem in \eqref{20020310_13} can be solved in closed form expression. The following theorem provides the solution.
\begin{theorem}\label{Theorem_vulnerability_singlesensor}
	The solution to the problem in~\eqref{20020310_13}, with $\Kc_o~=~\{1,2,\ldots,m\}$, is
	\begin{equation}\label{20220603_2}
	i= \argmin_{j \in \{1,2,\ldots,m\}} \left\{\left(\Sigmam_{Y\!Y}^{-1}\right)_{jj}\right\},
\end{equation}
	where $\Sigmam_{Y\!Y}$ is in~\eqref{20220706_1}.
\end{theorem}
\begin{proof}
	We start by noting that from~\eqref{20220915_1}, it yields that the vulnerability of measurement $i$ in~\eqref{def_Delta_i} is $\Delta(\textbf{0}, \lambda, v, i)$. From the equality in~\eqref{20220915_2}, the function $f (\textbf{0}, \lambda, 0, i)$ is a constant with respect to $i$. Hence, for $\Sigmam = \textbf{0}$, the optimization problem in~\eqref{20020310_13} is equivalent to
\begin{equation}\label{20220310_13}
	\min_{i \in \Kc_o} f(\textbf{0}, \lambda,v, i),
\end{equation}
	where $\Kc_o$ is defined in~\eqref{20220915_3}.
	Let $\lambda \in \mathbb{R}_+$ and $v \in \mathbb{R}_+$. From~\eqref{20220915_2}, the resulting problem in~\eqref{20220310_13} is equivalent to the following optimization problem:
	\begin{align} 
				\nonumber
		\min_{i \in \{1,2,\ldots,m\}}    & (1-\lambda)\textnormal{log}  \left| \Sigmam_{Y\!Y} + v\ev_i \ev_i^{\sf T} \right| -  \textnormal{log}  \left|  v\ev_i \ev_i^{\sf T} +\sigma^2 \textbf{I}_m\right| \\
		&+  \lambda v  \textnormal{tr}\left(\Sigmam_{Y\!Y}^{-1} \ev_i \ev_i^{\sf T} \right) \label{20220924_5}  \\
				\nonumber
		= \min_{i \in \{1,2,\ldots,m\}}& (1-\lambda)\textnormal{log}  \left| \textbf{I}_m + v\Sigmam_{Y\!Y}^{-1} \ev_i \ev_i^{\sf T} \right| -  \textnormal{log}  (v+ \sigma^2) \\
		&+  \lambda v  \textnormal{tr}\left(\Sigmam_{Y\!Y}^{-1} \ev_i \ev_i^{\sf T} \right) \label{20220924_6}  \\
				\nonumber
		= \min_{i \in \{1,2,\ldots,m\}}&(1 - \lambda)\textnormal{log}  \left( 1 + v \textnormal{tr}\left(\Sigmam_{Y\!Y}^{-1}\ev_i \ev_i^{\sf T} \right)\right) \\
		& +  \lambda v \textnormal{tr}\left(\Sigmam_{Y\!Y}^{-1}\ev_i \ev_i^{\sf T} \right), \label{20220924_7}
	\end{align} 
where the equivalence in~\eqref{20220924_5} holds from plugging in $\Sigmam = \textbf{0}$ into the equality in~\eqref{20220915_2}; the equality in~\eqref{20220924_6} follows from removing a constant $(1-\lambda)\textnormal{log} \left| \Sigmam_{Y\!Y} \right|$ from the first term; and the equality in~\eqref{20220924_7} follows from the fact that $\Sigmam_{Y\!Y}^{-1} \ev_i \ev_i^{\sf T}$ is a matrix with nonzero entries in the $i$-th column and all the other entries are zeros. 

We now proceed by defining $t \eqdef v \textnormal{tr}\left(\Sigmam_{Y\!Y}^{-1}\ev_i \ev_i^{\sf T} \right)$, with $t \in \mathbb{R}_+$. Hence, the equality in~\eqref{20220924_7} is equivalent to
\begin{align}\label{20221018_1}
\min_{t \in \mathbb{R}_+}&(1 - \lambda)\textnormal{log}  \left( 1 + t\right) + \lambda t.
\end{align} 
Note that~\eqref{20221018_1} is monotonically increasing with respect to $t$. Therefore, the cost function in~\eqref{20220924_7} is monotonically increasing with respect to $\textnormal{tr}\left(\Sigmam_{Y\!Y}^{-1}\ev_i \ev_i^{\sf T} \right)$. This completes the proof.
\end{proof}

From Theorem~\ref{Theorem_vulnerability_singlesensor}, it follows that the identification of the most vulnerable measurement is independent of $\lambda$, introduced in~\eqref{MI_f_k}, and the variance $v$. That is, it exclusively depends on the system topology and parameters denoted by $\Sigmam_{Y\!Y}$ defined in~\eqref{20220706_1}. This result coincides with Theorem~\ref{Theorem_singlesensor} in the sense that in the attack construction for $k=1$, the most vulnerable measurement is characterized in~\eqref{eq:opt_i}, which is independent of the value of $\lambda$. The following corollary formalizes this observation.

\begin{corollary}\label{coro_vulnerabilityordering_k0}
	Consider the parameters $\Sigmam = \textbf{0}$, $v \in \mathbb{R}_+$ and $\lambda \in \mathbb{R}_+$. The vulnerability ranking for all measurements 
	\begin{equation}
	\sv \eqdef (s_1, s_2, \ldots,s_m)
\end{equation}
	is such that for all measurements $i$, with~$i~\in~\{1,2,\ldots,m\}$, $s_i \in \{1,2,\ldots,m\}$ and 
	\begin{equation}
	\textnormal{tr}\left(\Sigmam_{Y\!Y}^{-1}\ev_{s_1} \ev_{s_1}^{\sf T} \right) \leq \textnormal{tr}\left(\Sigmam_{Y\!Y}^{-1}\ev_{s_2} \ev_{s_2}^{\sf T} \right) \leq \ldots \leq \textnormal{tr}\left(\Sigmam_{Y\!Y}^{-1}\ev_{s_m} \ev_{s_m}^{\sf T} \right).
\end{equation}
	For all $i\in \{1,2,\ldots,m\}$, the $i$-th most vulnerable measurement is $s_i$.
\end{corollary}

\subsection{Vulnerability index (VuIx)} 
The vulnerability analysis of uncompromised systems in Section~\ref{Vulnerability analysis on single sensor attack} is constrained to $k = 0$. To generalize the vulnerability analysis to compromised systems when $k > 0$, in the following we propose a novel metric, that is, {\it vulnerability index}, for all $i \in \Kc_o$.
\begin{definition}
	\label{def:vul_com}
	For $k \in \{1,2,\ldots,m-1\}$ and $\Sc_{k}$ in~\eqref{20220609_1}, consider the parameters $\Sigmam\in\Sc_{k}$, $v \in \mathbb{R}_+$, $\lambda \in \mathbb{R}_+$. Consider also the set $\{(i, \Delta): i \in \Kc_o \}$, with $\Kc_o$ in~\eqref{set_K_o} and 
	\begin{equation}
	\Delta_i \eqdef \Delta(\Sigmam, \lambda, v, i).
	\end{equation}
 Let the vulnerability ranking 
 \begin{equation}
 	\rv = (r_1, r_2, \ldots, r_{|\Kc_o|})
\end{equation}
 be such that for all $i \in \{1,2,\ldots,|\Kc_o|\}$, $r_i \in \Kc_o$ and moreover,
	\begin{equation}
	\Delta_{r_1}  \leq \Delta_{r_2} \leq \ldots \leq \Delta_{r_{|\Kc_o|}}. 
\end{equation}
	The vulnerability index (VuIx) of measurement $r_j \in \Kc_o$ is $j$, that is, $\textnormal{VuIx}(r_j)=j$.
\end{definition}

\begin{figure}
	\centering
	\includegraphics[width=0.5\textwidth]{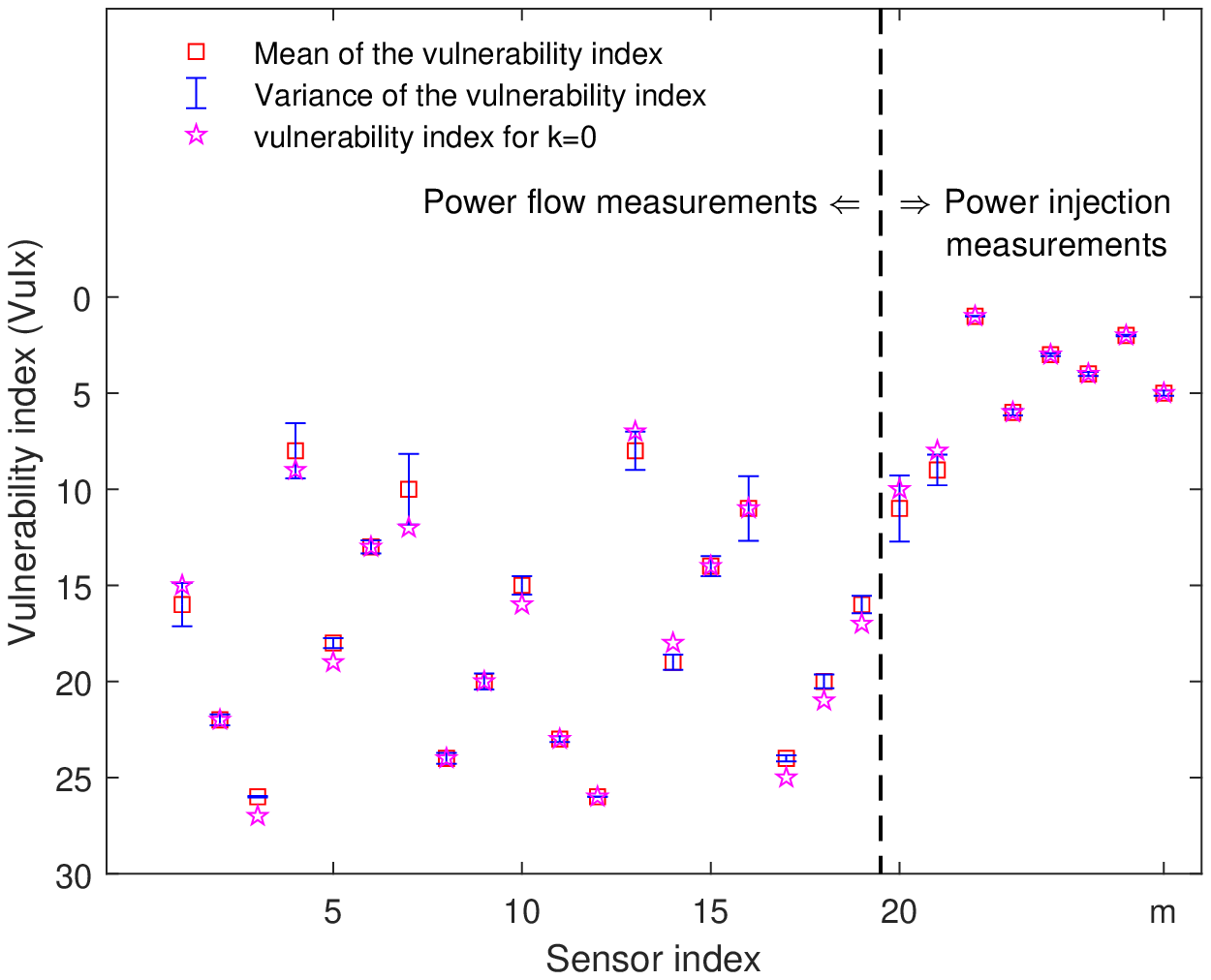}
	\caption{Vulnerability index (VuIx) when {$k = 1$, SNR = 10 dB}, $\lambda = 2$ and $\rho=0.1$ on the IEEE {9-bus system}.}\label{fig1}
	\centering
	\includegraphics[width=0.5\textwidth]{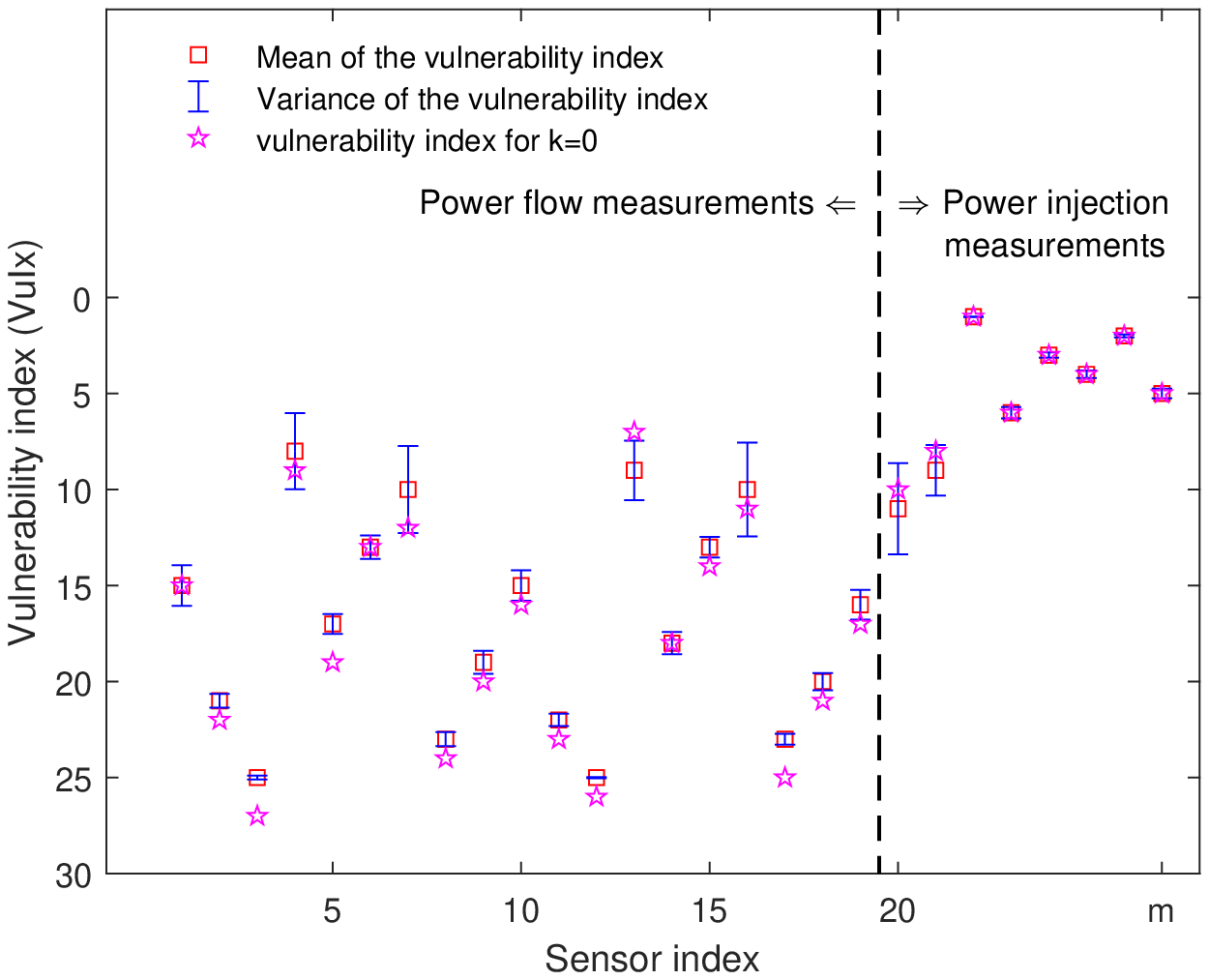}
	\caption{Vulnerability index (VuIx) when {$k = 2$, SNR = 10 dB}, $\lambda = 2$ and $\rho=0.1$ on the IEEE {9-bus system}.}\label{fig2}
\end{figure}
\begin{figure}
	\centering
	\includegraphics[width=0.5\textwidth]{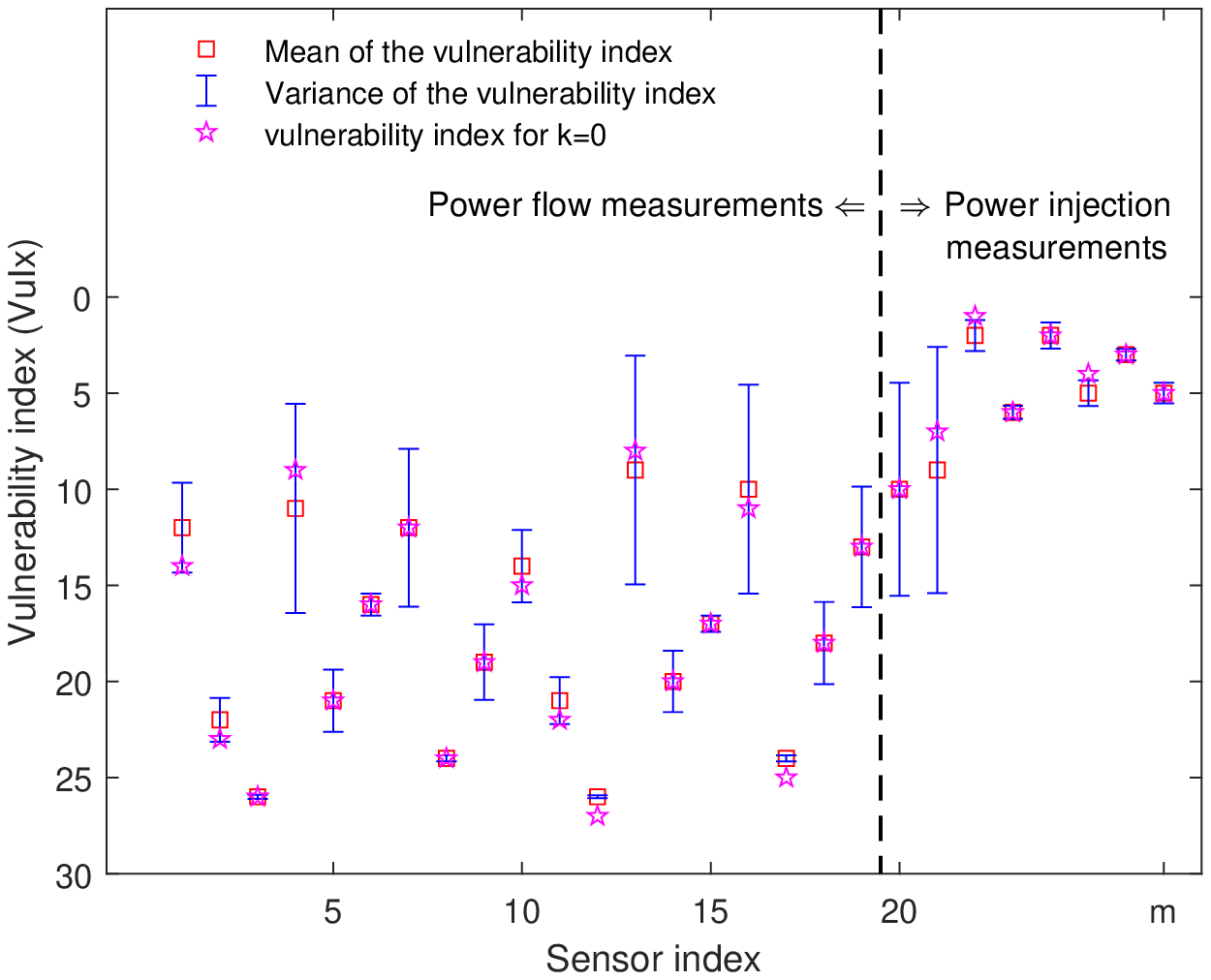}
	\caption{Vulnerability index (VuIx) when {$k = 1$, SNR = 30 dB}, $\lambda = 2$ and $\rho=0.1$ on the IEEE {9-bus system}.}\label{fig3}
	\centering
	\includegraphics[width=0.5\textwidth]{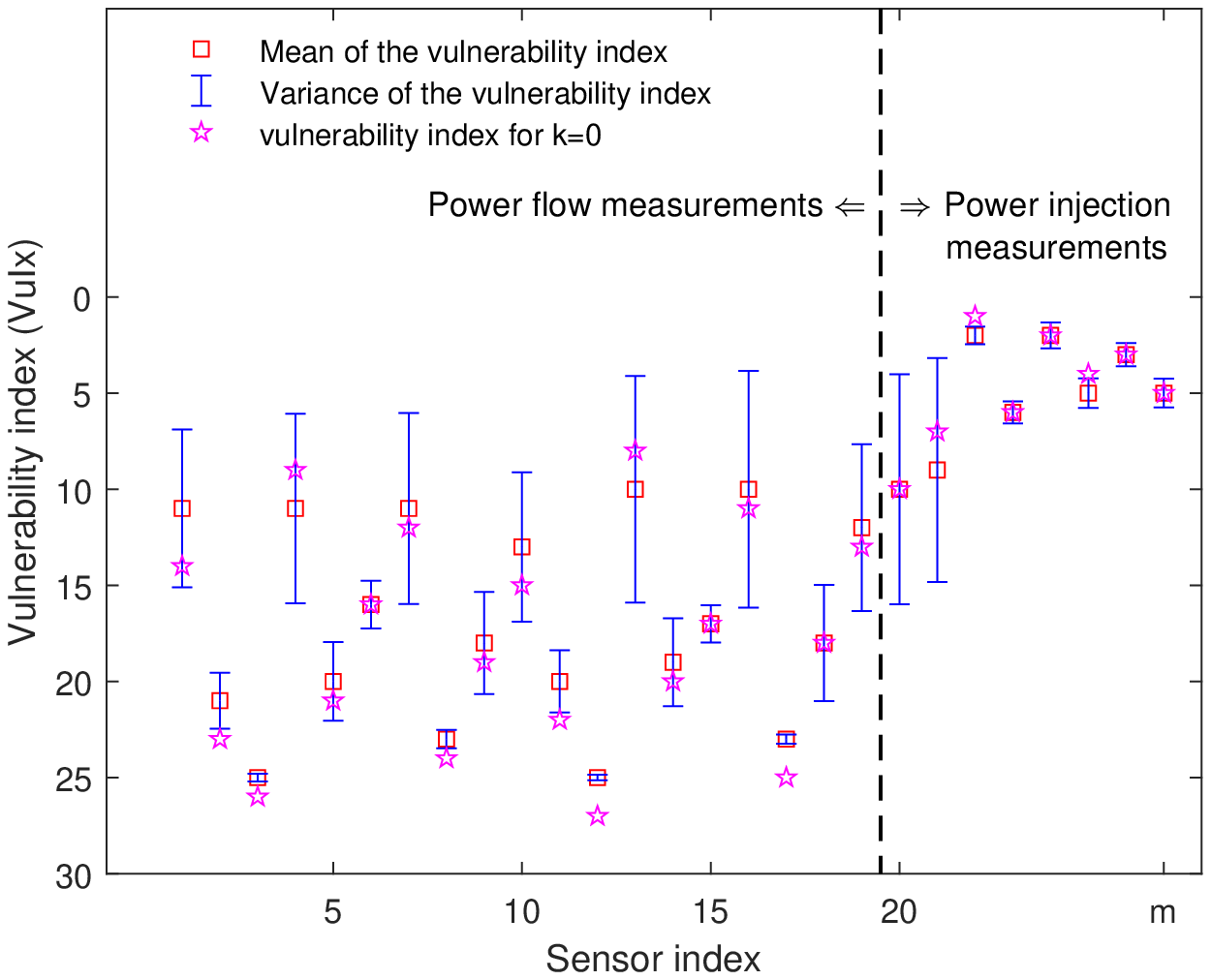}
	\caption{Vulnerability index (VuIx) when {$k = 2$, SNR = 30 dB}, $\lambda = 2$ and $\rho=0.1$ on the IEEE {9-bus system}.}\label{fig4}
\end{figure}

Note that the measurement with the smallest VuIx is the most vulnerable measurement that corresponds to the solution to the optimization problem in~\eqref{20020310_13}.
The proposed VuIx for $i \in \Kc_o$ is obtained from Algorithm~\ref{alg:vulnerability index}.
\begin{algorithm}
	\caption{Computation of Vulnerability Index (VuIx)}
	\label{alg:vulnerability index}
	\begin{algorithmic}[1]
		\Require $\Hm$ in \eqref{eq:obs_noattack};\newline
		$\sigma^2$ in \eqref{Z_distribution};\newline
		$\Sigmam_{X\!X}$ in \eqref{Sigma_xx};\newline
		$\Sigmam \in \Sc_k$ in \eqref{20220310_1};\newline
		$\lambda \in \mathbb{R}_+$ and 
		$v \in \mathbb{R}_+$.
		\Ensure the VuIx for all $i \in \Kc_o$.
		\State Set $\Kc_o$ in~\eqref{set_K_o}
		\For {$i \in \Kc_o$}
		\State Compute $\Delta(\Sigmam,\lambda,v,i)$ in~\eqref{def_Delta_i}
		\EndFor
		\State Sort $\Delta(\Sigmam, \lambda, v, i) $ in ascending order
		\State Set $\rv = (r_1, r_2, \ldots, r_{|\Kc_o|})$  
		\State Set the VuIx of measurement $r_j \in \Kc_o$ as $j$.
	\end{algorithmic}
\end{algorithm}

\begin{figure}[htbp]
	\centering
	\includegraphics[width=9cm]{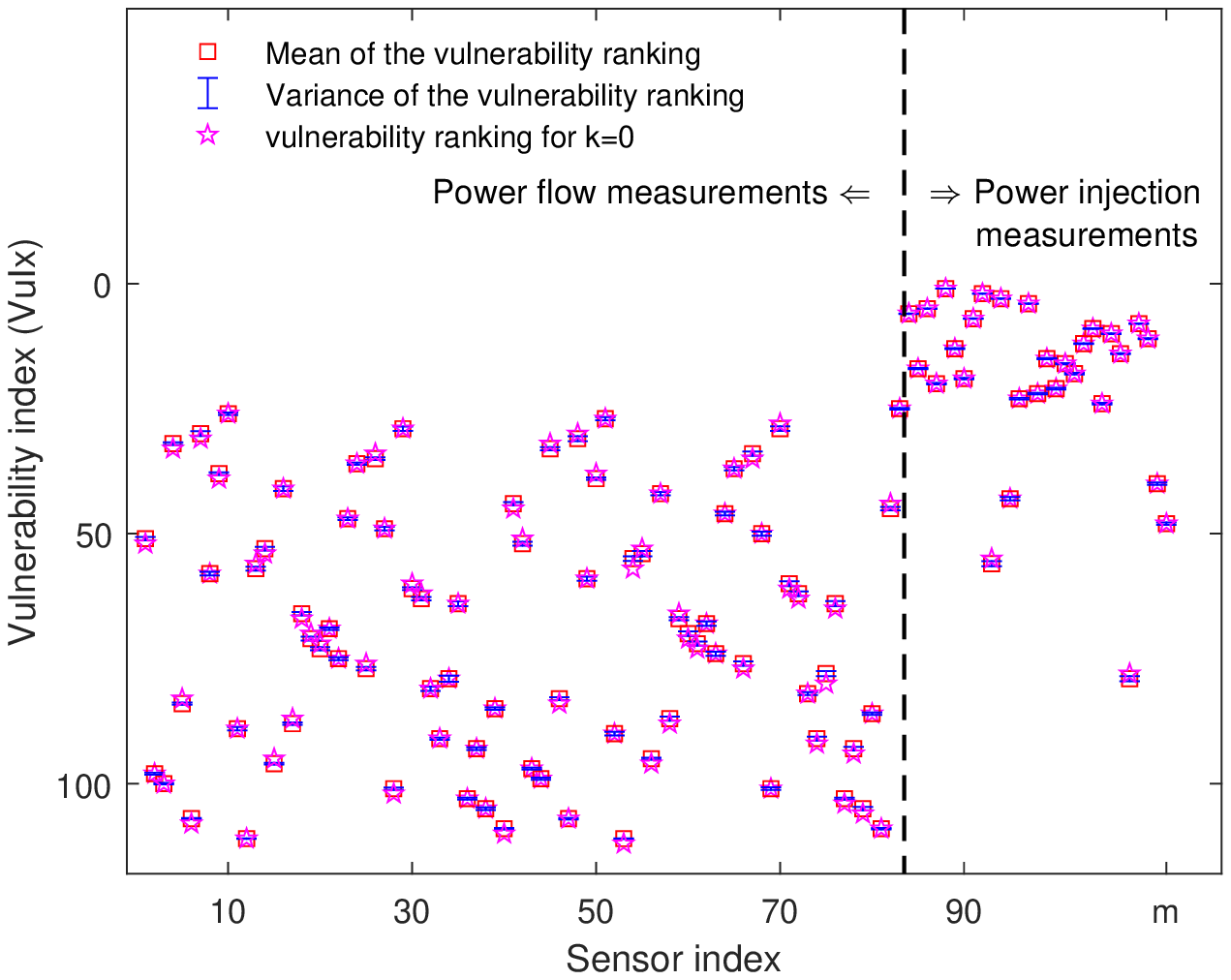}
	\caption{Vulnerability index (VuIx) when {$k = 1$}, SNR = 10 dB, $\lambda = 2$ and $\rho=0.1$ on the IEEE {30-bus system}.}\label{fig5}
	\centering
	\includegraphics[width=9cm]{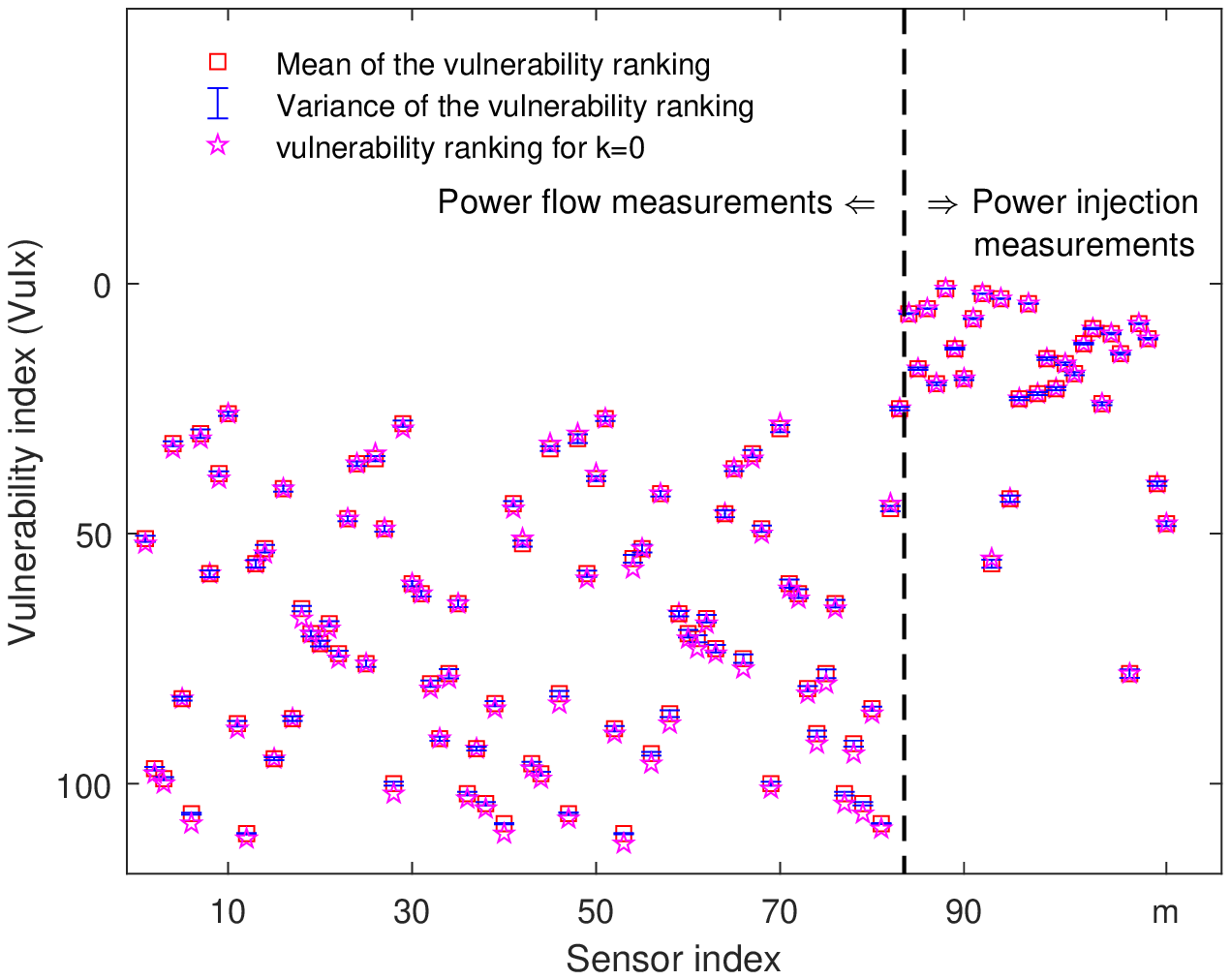}
	\caption{Vulnerability index (VuIx) when {$k = 2$}, SNR = 10 dB, $\lambda = 2$ and $\rho=0.1$ on the IEEE {30-bus system}.}\label{fig6}
\end{figure}
\begin{figure}[htbp]
	\centering
	\includegraphics[width=9cm]{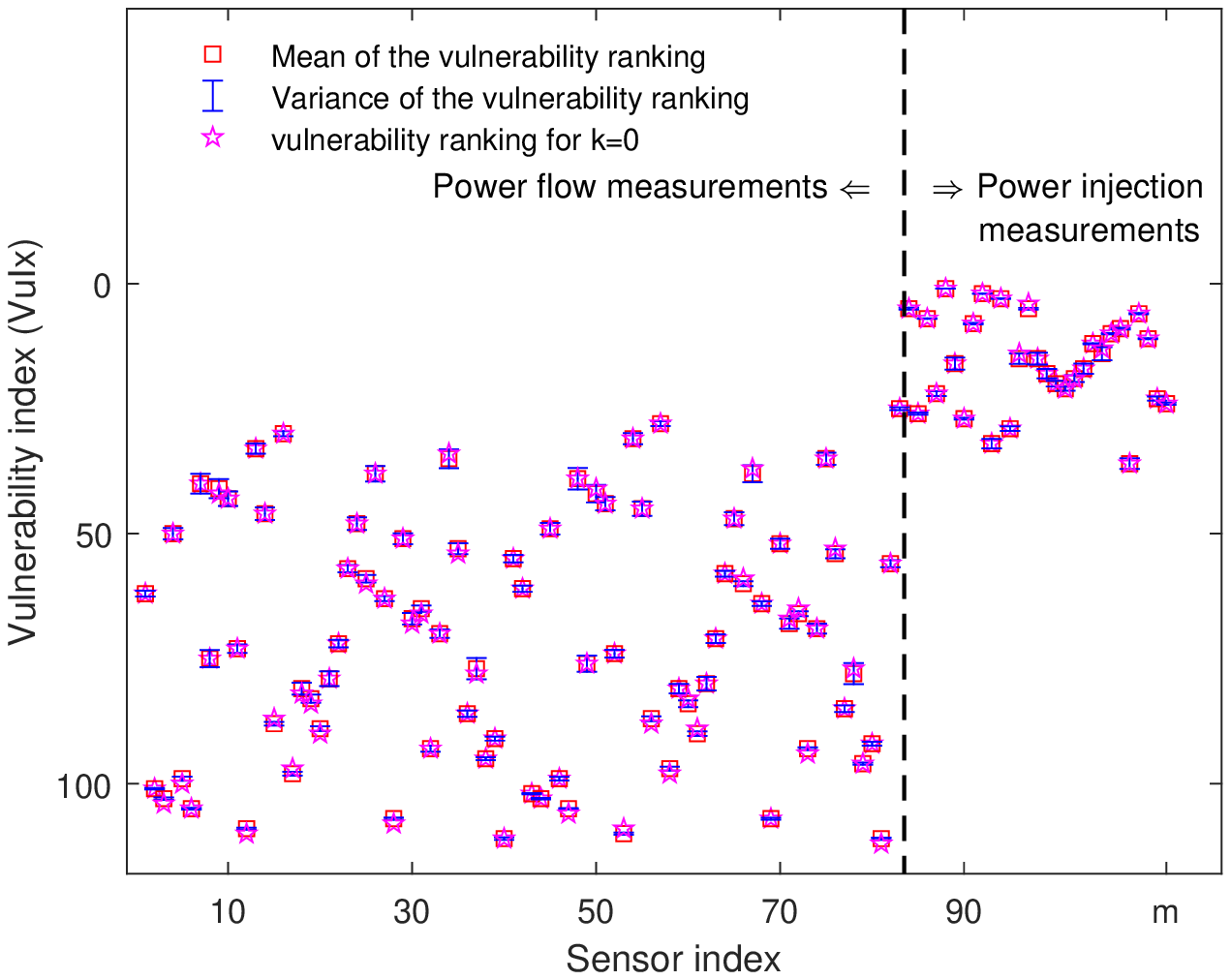}
	\caption{Vulnerability index (VuIx) when {$k = 1$}, SNR = 30 dB, $\lambda = 2$ and $\rho=0.1$ on the IEEE {30-bus system}.}\label{fig7}
	\centering
	\includegraphics[width=9cm]{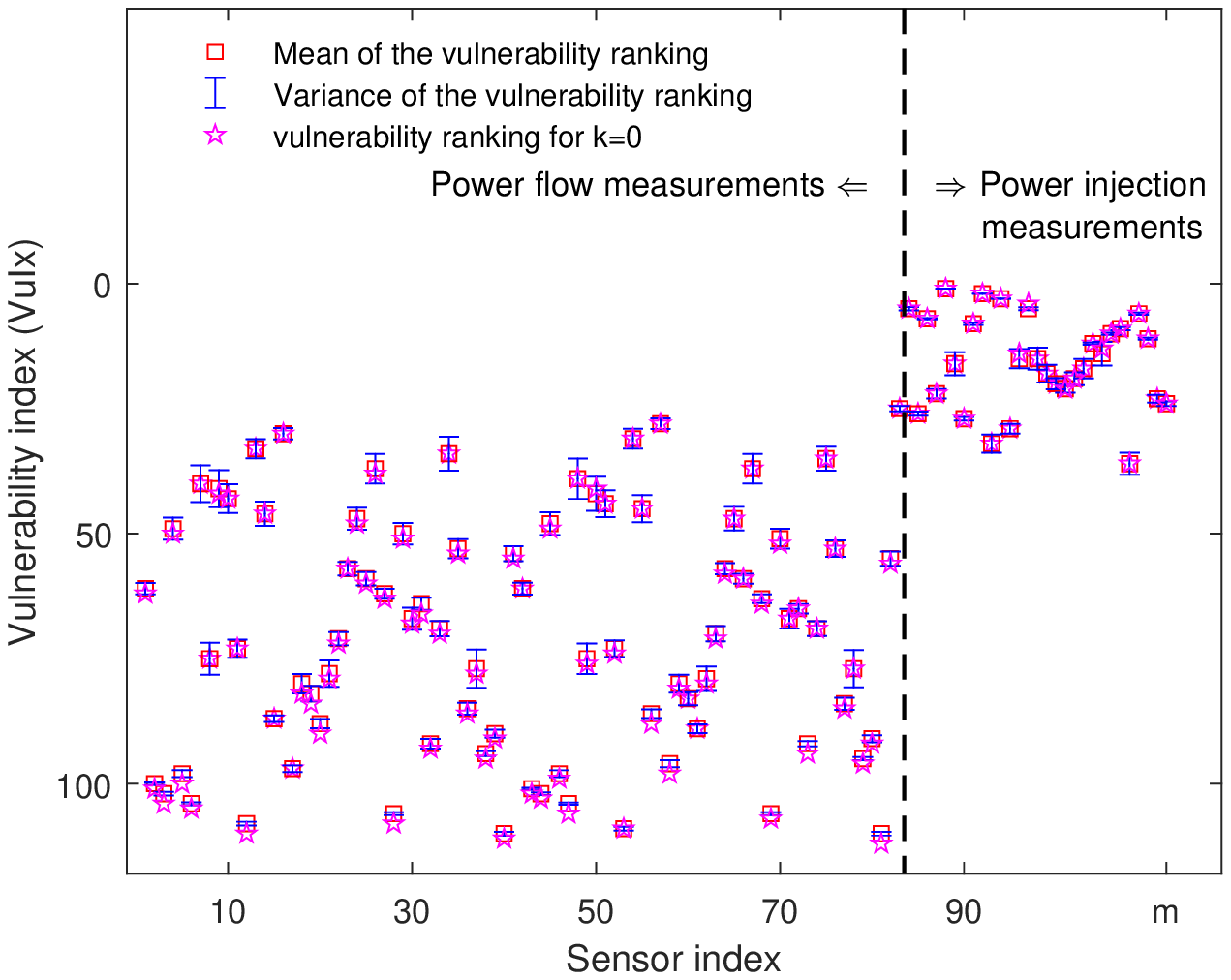}
	\caption{Vulnerability index (VuIx) when {$k = 2$}, SNR = 30 dB, $\lambda = 2$ and $\rho=0.1$ on the IEEE {30-bus system}.}\label{fig8}
\end{figure}
\vspace{-2mm}

\section{Numerical results}\label{sec_simulation}
In this section, we numerically evaluate the VuIx of the measurements on a direct current (DC) setting for the IEEE Test systems~\cite{UoW_ITC_99}. The voltage magnitudes are set to 1.0 per unit, that is, the measurements of the systems are active power flow between the buses that are physically connected and active power injection to all the buses. The Jacobian matrix $\Hm$ in~\eqref{eq:obs_noattack} determined by the topology of the system and the physical parameters of the branches is generated by MATPOWER~\cite{matpower}. We adopt a Toeplitz model for the covariance matrix $\Sigmam_{X\!X}$ that arises in a wide range of practical settings, such as autoregressive stationary processes. Specifically, we model the correlation between state variable $X_i$ and $X_j$ with an exponential decay parameter $\rho \in \mathbb{R}_+$, which results in the entries of the matrix $(\Sigmam_{X\!X})_{ij} = \rho^{|i-j|}$ with $(i,j) \in \{1,2,\ldots, n\} \times \{1,2,\ldots, n\}$.
In this setting, the VuIx of the measurements is also a function of the correlation parameter $\rho$, the noise variance $\sigma^2$, and the Jacobian matrix $\Hm$. The noise regime in the observation model is characterized by the signal to noise ratio (SNR) defined as
\begin{equation}
\textrm{SNR} \eqdef 10\log_{10}\left(\frac{\textnormal{tr}(\Hm\Sigmam_{X\!X}\Hm^{\sf T})}{m\sigma^{2}}\right).
\end{equation}

For all $\lambda \in \mathbb{R}_+$ and $v \in \mathbb{R}_+$, we generate a realization of $k$ attacked indices $\Kc_a\subseteq\{1, 2, \ldots, m\}$ that is uniformly sampled from the set of sets given by
\begin{equation}
\tilde{\Kc}=\left\{\Ac\subseteq\{1, 2, \ldots, m\}: |\Ac|=k\right\}.
\end{equation}
We then construct a random covariance matrix describing the existing attacks on the system as
\begin{equation}
\tilde{\Sigmam}=\sum_{i\in{\Kc_a}}\ev_i\ev_i^{\sf T},
\end{equation}
with $\Kc_a\in\tilde{\Kc}$.
In the numerical simulation, we obtain the vulnerability of measurement $i$ by computing
\begin{equation}
\Delta(\tilde{\Sigmam}, \lambda, 1, i),
\end{equation}
where $i \in \Kc_o$ is in~\eqref{set_K_o} and $\Delta$ is defined in~\eqref{def_Delta_i}.

\subsection{Assessment of vulnerability index (VuIx)}\label{Performance of VuIx for all the measurements}
Fig.~\ref{fig1} and Fig.~\ref{fig2} depict the mean and variance of the VuIx obtained from Algorithm~\ref{alg:vulnerability index} for all the measurements with SNR = 10 dB, $\lambda = 2$ and $\rho = 0.1$ on the IEEE 9-bus system when $k = 1$ and $k =2$, respectively. 
It is observed that power injection measurements yield higher priority vulnerability indices, which indicates that power injection measurements are more vulnerable to data integrity attacks. 
Most power injection measurements correspond to higher ranked vulnerability indices but there are instances of power flow measurements with a higher ranked VuIx than that of some power injection measurements. 
Interestingly, the power injection measurements with lower vulnerability indices correspond to the buses that are more isolated in the system, that is, the buses with a lower number of connections. On the other hand, the power flow measurements with higher ranked vulnerability indices correspond to the branches with higher admittance. 
The VuIx for $k=0$ obtained in Corollary~\ref{coro_vulnerabilityordering_k0} is depicted for the purpose of  serving as a reference to assess the deviation when $k>0$. Interestingly, the VuIx of most measurements does not change substantially for different values of $k$, which suggests that the VuIx is insensitive to the state of the system.

Fig.~\ref{fig3} and Fig.~\ref{fig4} depict the mean and variance of the VuIx from Algorithm~\ref{alg:vulnerability index} for all the measurements with SNR = 30 dB, $\lambda = 2$ and $\rho = 0.1$ on the IEEE 9-bus system when $k = 1$ and $k = 2$, respectively. Interestingly, the mean of the VuIx for most of the measurements does not deviate significantly from the case when $k = 0$. Instead, most of the variances deviate significantly in comparison with the cases in Fig.~\ref{fig1} and Fig.~\ref{fig2} with SNR = 10 dB.
Fig.~\ref{fig5} and Fig.~\ref{fig6} depict the results on IEEE 30-bus systems with the same setting as in Fig.~\ref{fig1} and Fig.~\ref{fig2}, respectively. Fig.~\ref{fig7} and Fig.~\ref{fig8} depict the results on IEEE 30-bus systems with the same setting as in Fig.~\ref{fig3} and Fig.~\ref{fig4}, respectively. Surprisingly, the mean of the VuIx in larger systems coincides with that obtained for the case $k=~0$, which shows that the VuIx is a robust security metric for large systems. Interestingly, the power injection measurements corresponding to the least connected buses decrease in the VuIx when SNR = 10 dB. 
\subsection{Comparative vulnerability assessment of power flow and power injection measurements}
In Section~\ref{Performance of VuIx for all the measurements} we have established that power injection measurements and power flow measurements are qualitatively different in terms of the VuIx. 
To provide a quantitative description of this difference, Fig.~\ref{fig9} depicts the probability of a given VuIx $i\in\left\{1, 2, \ldots, m-|\Kc_a|\right\}$ being taken by a power injection measurement or a power flow measurement for the IEEE 9-bus and 30-bus systems when $\lambda = 2$, $k = 2$, SNR = 30 dB and $\rho=0.1$. Specifically, Fig.~\ref{fig9} depicts the probability of the following events:
\vspace{-2mm}
\begin{align*}
	{\sf Flow}_i\!&\!: \textnormal{VuIx $i$ corresponds to a power flow measurement},\\
	{\sf Inj}_i\!&\!: \textnormal{VuIx $i$ corresponds to a power injection measurement}.
\end{align*}
It is observed that in both systems, small VuIx are more likely to correspond to power injection measurements than to power flow measurements, that is, $\mathbb{P}[{\sf Inj}_i]>\mathbb{P}[{\sf Flow}_i]$ for small values of $i$. Conversely, it holds that $\mathbb{P}[{\sf Inj}_i]<\mathbb{P}[{\sf Flow}_i]$ for large values of $i$.
In fact, small VuIx corresponding to power injection measurements is with probability one, which shows that the most vulnerable measurements in the system are always power injection measurements. 
Conversely, the larger VuIx values corresponding to power flow measurements is with probability one, which indicates that the least vulnerable measurements are always power flow measurements. 
\begin{figure}[htbp]
	\vspace{-1mm}
	\centering
	\includegraphics[width=9cm]{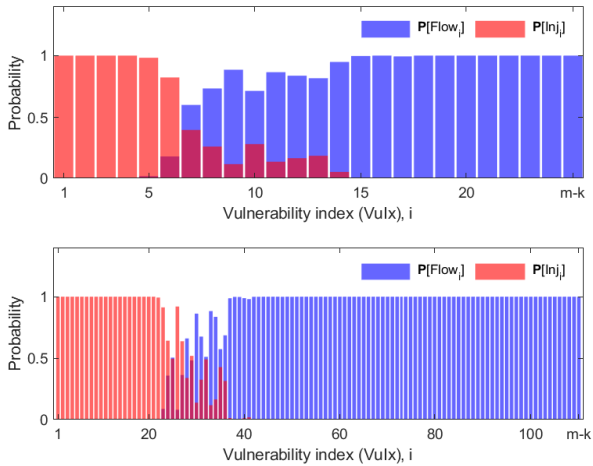}
	\caption{Probability mass function of Vulnerability index (VuIx) for power injection measurements and power flow measurements when $\lambda = 2$, $k = 2$, SNR = 30 dB and $\rho=0.1$ on IEEE 9-bus and 30-bus systems, respectively.}\label{fig9}
\end{figure}
Interestingly, there is a clear demarcation for each system for which $\mathbb{P}[{\sf Inj}_i]$ and $\mathbb{P}[{\sf Flow}_i]$ change rapidly with the VuIx value, which suggests a phase transition type phenomenon for measurement vulnerability.

The probability of VuIx taken by power injection measurements has high probability mass in higher priority vulnerability indices. One the other hand, power flow measurements are with higher probability mass in low ranked VuIx. Precisely, the probability of the vulnerability indices with higher priority taken by power injection measurements is 1 in both IEEE 9-bus and 30-bus systems. Meanwhile, the probability of the lower ranked vulnerability indices taken by power flow measurements is 1.
Note that the probability of medium ranked vulnerability indices taken by power injection measurements drops significantly, which indicates that there are some power flow measurements that are equally as vulnerable as power injection measurements. We observe that these power flow measurements correspond to the branches with higher admittance. The power injection measurements with lower vulnerability indices correspond with the buses that are isolated in the systems.

Fig.~\ref{fig10} depicts the distribution of VuIx for power injection measurements and power flow measurements on the IEEE 9-bus and 30-bus systems when $\lambda = 2$, $k = 2$, SNR = 30 dB and $\rho=0.1$. Specifically, Fig.~\ref{fig10} depicts the probability mass function of the following events:
\vspace{-2mm}
\begin{align*}
		{\sf VuIx(Flow)} = i\!&\!: \textnormal{VuIx for power flow measurements is $i$},\\
		{\sf VuIx(Inj)} = i\!&\!: \textnormal{VuIx for power injection measurements is $i$}.
\end{align*}
Power injection measurements have much higher probability with high ranked VuIx. VuIx, whereas power flow measurements have much higher probability with low ranked VuIx. It is worth noting that the probability mass functions are close to uniform for high and low vulnerability index ranges. This suggests that the most vulnerable measurements in the system are contained with high probability in a subset of the power injection measurements. Conversely, the least vulnerable measurements comprise the majority of the power flow measurements with no apparent preference over the majority. Interestingly, in the 30-bus system, the probability of lowest ranked VuIx for power flow measurements experiences a sharp increase.

\section{Conclusion}\label{sec_conclusion}
In this paper, we have proposed, from a fundamental perspective, a novel security metric referred to as
vulnerability index (VuIx) that characterizes the vulnerability of power system measurements to data integrity attacks.
We have achieved this by embedding information theoretic measures into the metric definition. The resulting VuIx framework evaluates the vulnerability of all the measurements in the systems and enables the operator to identify those that are more exposed to data integrity threats. We have tested the framework for IEEE test systems and concluded that power injection measurements are more vulnerable to data integrity attacks than power flow measurements. 

\begin{figure}[htbp]
	\centering
	\includegraphics[width=9cm]{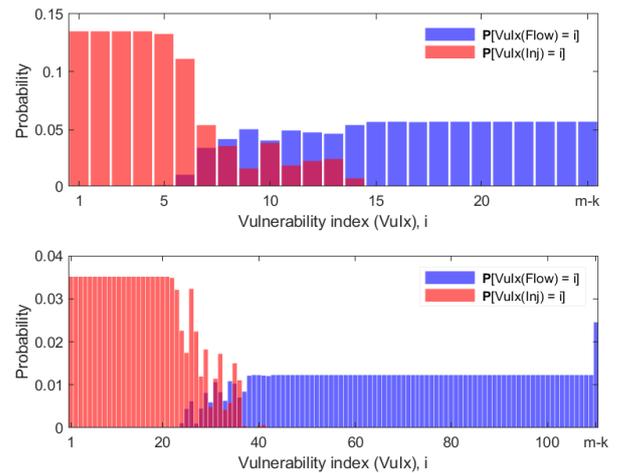}
	\caption{Probability of Vulnerability index (VuIx) corresponds to power injection measurements and power flow measurements when $\lambda = 2$, $k = 2$, SNR = 30 dB and $\rho=0.1$ on IEEE 9-bus and 30-bus systems, respectively.}\label{fig10}
\end{figure}
\bibliography{IET_SG_22}
\bibliographystyle{vancouver}
\end{document}